%% file: main.tex
\documentclass[authorversion, nonacm, sigconf, screen]{acmart}
\acmConference[ICAIF '25]{6th ACM International Conference on AI in Finance}{November 15--18, 2025}{Singapore, Singapore}
\acmBooktitle{6th ACM International Conference on AI in Finance (ICAIF '25), November 15--18, 2025, Singapore, Singapore}\acmDOI{10.1145/3768292.3770412}
\acmISBN{979-8-4007-2220-2/2025/11}
\usepackage{amsmath, amsthm, dsfont}
\newtheorem{assumption}{Assumption}
\newtheorem{proposition}{Proposition}

\usepackage{caption}
\usepackage{subcaption}
\usepackage{mathtools}
\mathtoolsset{showonlyrefs=true}
\usepackage{bm}

\usepackage{diagbox}
\usepackage{comment}

\usepackage{algorithm}
\usepackage{algpseudocode}
\algrenewcommand{\algorithmicreturn}{\textbf{Return}}
\algnewcommand{\Returnx}[1]{\Statex \algorithmicreturn\ #1} 

\newtheorem{remark}{Remark}%

\newcommand{\indfun}{\mathds{1}}

\usepackage{ifthen}
\newboolean{anonym}
\setboolean{anonym}{false}

\definecolor{colour3}{RGB}{178,55,250} % purple

\newcounter{noteMCctr} 

\definecolor{darkpastelgreen}{rgb}{0.01, 0.75, 0.24}

\begin{document}

\title{Factor-Driven Network Informed Restricted Vector Autoregression}

% 

\author{Brendan Martin}
\affiliation{%
  \institution{Imperial College London}
  \country{United Kingdom}
}
\email{b.martin22@imperial.ac.uk}
%\orcid{0009-0007-2168-8261}

\author{Mihai Cucuringu}
\affiliation{%
  \institution{University of California, Los Angeles}
  \country{United States}
}

\additionalaffiliation{%
  \department{Department of Statistics}
  \institution{University of Oxford}
  \city{Oxford}
  \country{United Kingdom}
}
\additionalaffiliation{%
  \department{Oxford-Man Institute of Quantitative Finance}
  \institution{University of Oxford}
  \city{Oxford}
  \country{United Kingdom}
}
\email{mihai@math.ucla.edu}

%\orcid{0000-0002-8464-2152}

\author{Alessandra Luati}
\affiliation{%
  \institution{Imperial College London}
  \country{United Kingdom}
}
\additionalaffiliation{%
  \department{Department of Statistical Sciences}
  \institution{University of Bologna}
  \city{Bologna}
  \country{Italy}
}
\email{a.luati@imperial.ac.uk}

%\orcid{0000-0001-6407-9385}

\author{Francesco Sanna Passino}
\affiliation{%
  \institution{Imperial College London}
  \country{United Kingdom}
}
\email{f.sannapassino@imperial.ac.uk}

%\orcid{0000-0002-4571-6681}

%% \renewcommand{\shortauthors}{Brendan Martin, Mihai Cucuringu, Alessandra Luati, and Francesco Sanna Passino}

\begin{abstract}
 \input{Sections/abstract}
\end{abstract}

%% Please copy and paste the code.
%%
%%% CCS classification

\keywords{Factor models, financial returns prediction, network time series, stochastic block model.}


\maketitle

\input{Sections/introduction}

\input{Sections/model}

\input{Sections/estimation.tex}

\input{Sections/simulation1}

\section{Empirical analysis}\label{sec:applications}

The FNIRVAR model is applied to three data sets: \textit{(i)} a financial dataset of daily excess market returns; \textit{(ii)} a financial dataset of thirty minutely intra-day excess market returns; \textit{(iii)} a vintage of macroeconomic indicators typically used for predicting US industrial production. Descriptions of the data sources and transformations applied to each of the raw data sets are given below. Our primary interest is the one-step-ahead out-of-sample predictive performance of FNIRVAR. We compare the results against the factor plus sparse VAR model of \citet{miao2023high} and \citet{fan2023bridging} as well as a static factor model with no idiosyncratic model. \citet{fan2023bridging} estimate the sparse VAR model using LASSO whereas \citet{krampe2025factor} estimate it using adaptive LASSO. We follow \citet{fan2023bridging} and use LASSO, with the tuning parameter selected via BIC. For the financial applications, we examine the economic benefit of an FNIRVAR trading strategy under different portfolio allocation and transaction cost settings. 
\input{Sections/daily_returns}

\input{Sections/intraday_returns}
\input{Sections/FRED-MD}

\input{Sections/conclusion} 

\begin{acks}
BM acknowledges funding from the Engineering and Physical Sciences Research Council (EPSRC), grant number EP/S023151/1.
FSP acknowledges funding from the EPSRC, grant no. EP/Y002113/1.
\end{acks}

\bibliographystyle{ACM-Reference-Format}
\bibliography{references}


\end{document}

%% file: Sections/abstract.tex
High-dimensional financial time series often exhibit complex dependence relations driven by both common market structures and latent connections among assets. To capture these characteristics, this paper proposes Factor-Driven Network Informed Restricted Vector Autoregression (FNIRVAR), a model for the common and idiosyncratic components of high-dimensional time series with an underlying unobserved network structure. The common component is modelled by a static factor model, which allows for strong cross-sectional dependence, whilst a network vector autoregressive process captures the residual co-movements due to the idiosyncratic component. An assortative stochastic block model underlies the network VAR, leading to groups of highly co-moving variables in the idiosyncratic component. For estimation, a two-step procedure is proposed, whereby the static factors are estimated via principal component analysis, followed by estimation of the network VAR parameters. %The method allows for estimation when the underlying network is unobserved, which is the setting considered here. 
%The efficacy of the 
The method is demonstrated in financial applications to daily returns, intraday returns, and FRED-MD macroeconomic variables. In all %three 
cases, the proposed method outperforms a static factor model, as well as a static factor plus LASSO-estimated sparse VAR model, in terms of forecasting and financial performance metrics. 

%% file: Sections/introduction.tex
\section{Introduction} \label{sec:intro}

Datasets in which the number of variables is comparable to or greater than the number of observations are ubiquitous in finance and economics. In a high-dimensional time series setting, correctly characterizing the dependence among variables is crucial for interpreting the relationships between series as well as for prediction tasks. Two main techniques for imposing low-dimensional structure on high-dimensional time series have emerged in the literature \citep[][]{fan2023bridging,giannone2021economic,chernozhukov2017lava}. First, factor models are used to capture co-movements due to shocks common across all series \citep[see][and references therein]{forni2009opening}. Second, sparse models, whereby each series depends only on a small subset of the other series, have proven effective tools for parameter estimation and forecasting \citep[see, for example,][]{fan2020factor}. 
In fact, any second-order stationary process $\{X_{i,t} : i \in \mathbb{N}, t \in \mathbb{Z}\}$ can be decomposed into a so-called common or systematic component, driven by a finite number of (possibly dynamic) factors, and an orthogonal idiosyncratic component \citep{barigozzi2024dynamic}. As such, combining dense factor models with sparse methods has been widely used %in the literature 
to capture both sources of co-movement \citep{fan2023bridging,miao2023high,krampe2025factor,barigozzi2024fnets,chen2023community}. Dense, here, refers to any model in which all %possible 
explanatory variables play a role in prediction, even though the impact of some of them may be small. %Ridge regression \citep[][]{hoerl1970ridge} is another example of dense statistical modelling.

In this work, we propose a novel factor-based framework for modelling multivariate time series with an underlying network structure, called Factor-Driven Network Informed Restricted Vector Autoregression (FNIRVAR). The proposed approach combines a static factor model \citep{stock2002forecasting,bai2003inferential} with a network-based vector autoregressive (VAR) model for the idiosyncratic component. Network VAR models,  introduced by \citet{zhu2017network} and \citet{knight2016modelling}, treat each univariate time series as being observed on a vertex of a graph, with the VAR coefficient matrix encoding the edge structure of the graph. The network structure can be thought of as a variable selection mechanism, with only those vertices which are connected having non-zero VAR coefficients. This regularisation allows for parameter estimation in high dimensions. The majority of existing works on network VAR methods assume the network to be observable \citep[see, for example,][]{knight2020generalized,barigozzi2025factor,yin2023general}. In financial and economic applications, however, the assumption of an observed underlying network is often unrealistic. A recent framework proposed by \citet{martin2024nirvar}, called Network Informed Restricted Vector Autoregression (NIRVAR), allows for an unobserved network, instead using a data-driven method to estimate the network. We therefore adopt the NIRVAR methodology to model the idiosyncratic component. Since we augment the NIRVAR model by including unobserved factors, we denote this factor driven model as FNIRVAR.

For the estimation of FNIRVAR, we propose a two-step procedure consisting in factor estimation based on principal component analysis (PCA), followed by NIRVAR estimation on the remaining component. 
The NIRVAR estimation framework entails finding groups of highly co-moving time series via spectral embedding and clustering, and subsequently estimating the VAR coefficients for series in the same group only \citep{martin2024nirvar}. All VAR parameters for variables belonging to different groups are restricted to zero. In a financial prediction context, this grouping method could be construed as a generalisation of a pairs trading strategy \citep[see][for a discussion of pairs trading]{vidyamurthy2004pairs}: instead of building a strategy based upon two highly correlated assets, we compare groups of many highly correlated assets. In particular, the future returns of a given asset are forecasted as a weighted sum of the assets within its group (after controlling for common factors).
In our proposed estimation procedure, we follow the approach of \citet{fan2023bridging} by assuming a large eigengap between the covariance matrix of the common component and that of the idiosyncratic component. An alternative approach is the joint optimisation method of \citet{agarwal2012noisy}. This, however, requires a sparsity assumption on the idiosyncratic component which we do not impose. In Section \ref{sec:sim}, we conduct a simulation study showing that assuming $r$ strong factors in the FNIRVAR model leads to the largest (in magnitude) $r$ eigenvalues exploding with dimension $N$, whereas the idiosyncratic eigenvalues are bounded. This justifies our proposed two-step estimation procedure. 

\subsection{Main contributions}
The primary contribution of this paper is the development of FNIRVAR, a factor driven NIRVAR model, and a related estimation framework. To the best of our knowledge, this is the first study that combines a dense factor model with a sparse network time series model without assuming the network to be observed. The benefits of this approach are three-fold: \textit{(i)} accounting for unobserved common factors mitigates endogeneity issues that arise if the common component is not modelled; \textit{(ii)} the NIRVAR estimation method captures residual co-movements between groups of series which are not picked up by the estimated factors, leading to enhanced forecasting performance; \textit{(iii)} the estimated factors can lead to meaningful interpretations about the forces driving the common component whilst estimated NIRVAR groups can shed light on the network driving the idiosyncratic component. The usefulness of the proposed methodology is demonstrated in an extensive application to financial returns prediction, as well as in a prediction task on the macroeconomic FRED-MD dataset \citep[][]{mccracken2016fred}.

\subsection{Outline}
The rest of this paper is organised as follows: the remainder of Section~\ref{sec:intro} discusses related literature and establishes the notation used in the paper. 
Next, Section \ref{sec:model} defines the FNIRVAR model and gives the conditions under which the model is stationary. Section \ref{sec:estimation} states the strong factor assumption with a large eigengap, and details the two-step estimation procedure. Simulation studies analysing the predictive performance of FNIRVAR under various data generating processes, and the effect of violating the large eigengap assumption on FNIRVAR estimation are illustrated in Section \ref{sec:sim}. Section \ref{sec:applications} contains applications to financial returns prediction and a macroeconomic application to forecasting US industrial production using FRED-MD. 
% The financial application is further subset \MC{?} into daily returns prediction and intraday returns prediction. 
All code is available in the GitHub repository \ifthenelse{\boolean{anonym}}{{\color{red}ANONYMISED LINK}}{\href{https://github.com/bmartin9/fnirvar}{bmartin9/fnirvar}}. 

\subsection{Related literature} 
This work contributes to the growing literature combining dense factor models with sparse models \citep[see, for example,][]{fan2023bridging,krampe2025factor,miao2023high,barigozzi2024fnets,chen2023community} and references therein. Factor models have a long history in econometrics starting with the dynamic exact factor models of \citet{geweke1977dynamic} and \citet{sargent1977business}. Here, exact refers to the assumption of an idiosyncratic error term with diagonal covariance, whereas approximate factor models have cross-sectional dependence in the idiosyncratic term. In particular, the factor model structure considered in this work is the static approximate factor model, introduced by \citet{chamberlain1983funds} and \citet{chamberlain1983arbitrage}, and further developed by \citet{stock2002forecasting}, \citet{stock2002macroeconomic}, \citet{bai2003inferential}, and \citet{bai2002determining}. A wider class that encompasses all other factor models is the General Dynamic Factor Model proposed by \citet{forni2000generalized} and \citet{forni2001generalized}. We note that although we impose dynamics on the factors in this work, we require a finite number of lags, and any dynamic factor model with a finite number of lags can be written as a static factor model \citep[][]{barigozzi2024dynamic}. For this reason, we consider our contribution as a static approximate factor model.

As an alternative to factor models, regularised estimation procedures such as the Least Absolute Shrinkage and Selection Operator \citep[LASSO;][]{tibshirani1996regression} impose sparsity assumptions to achieve dimensionality reduction. In the context of VAR modelling, there is a large literature on regularised estimation of sparse VAR models \citep[see, for example][]{basu2015regularized,han2015direct,knight2016modelling,adamek2023lasso,nicholson2020high}. 

The approaches of factor and sparse models have been first combined in the seminal paper of \citet{fan2023bridging}. Next, both \citet{miao2023high} and \citet{krampe2025factor} propose $\ell_{1}$-type regularised estimators for a factor-augmented sparse VAR model. Similar to our framework, \citet{krampe2025factor} utilise a two-step estimation procedure in which PCA is used in a first step to estimate the factors and the idiosyncratic component is estimated in a second step. In contrast, \citet{miao2023high} employ a joint optimisation procedure to estimate the common and idiosyncratic components. 

Additionally, there are a number of papers that develop network-based models for the idiosyncratic component, whilst including a factor model to account for the common component. \citet{chen2023community} extends the network autoregressive model of \citet{zhu2017network} to allow group-wise network effects as well as static factors. The major difference between our work and that of \citet{chen2023community} is that we do not assume that the network is observed. Furthermore, their model is based on an expected adjacency matrix, whereas the model we consider is based on a realised adjacency matrix. \citet{barigozzi2024fnets} develop a dynamic factor plus sparse VAR model and propose an $\ell_{1}$-regularised Yule-Walker estimator for the latent VAR process. The variables selected by their $\ell_{1}$-regularised method are then used to construct a network. In contrast, we construct the network as a first step, and use it to inform the estimation of the latent VAR process.

Finally, we mention the growing use of modern machine learning techniques for time series prediction tasks. \citet{gu2020empirical} provide a comprehensive overview and comparison of machine learning methods for empirical asset pricing, and suggest that the inclusion of nonlinear predictor interactions can lead to predictive gains missed by more traditional regression-based strategies. 

%\vspace{-2mm}
\subsection{Notation}
For an integer $k$, we let $[k]$ denote the set $\{1,\dots,k\}$. The $n \times n$ identity matrix is written $I_{n}$, and the indicator function is $\indfun\{\mathcal B\} = 1$ if the event $\mathcal{B}$ occurs and $0$ otherwise. For vectors, $\bm{v}_{1},\dots,\bm{v}_{n} \in \mathbb{R}^{p}$, we let $M = (\bm{v}_{1},\dots,\bm{v}_{n}) \in \mathbb{R}^{p \times n}$ be the matrix whose columns are given by $\bm{v}_{1},\dots,\bm{v}_{n}$. The transpose of $M$ is denoted $M^{\prime}$. We write $M_{i,:}$ and $M_{:,j}$ to denote the $i$-th row and $j$-th column of $M$ considered as vectors, respectively. We use $\odot$ to denote the Hadamard product between two matrices of the same dimensions. For matrices $M_{1},\dots,M_{r}$, $(M_{1}|\cdots|M_{r})$ and $(M_{1};\dots;M_{r})$ denote the column-wise and row-wise concatenation of the matrices, respectively. We write $\mathcal{S}_{K-1} = \{x \in \mathbb{R}^{K} : \sum_{k=1}^{K}x_{k} = 1, x_{k} > 0, k \in [K]\}$ to denote the interior of the unit simplex in $K$ dimensions.

%% file: Sections/model.tex
\section{Model}\label{sec:model}

In this section, we describe the FNIRVAR model, the main contribution of this work.
In order to define FNIRVAR, we first introduce the stochastic block model, which is the random graph used to model the idiosyncratic component. For a review of stochastic block models, see \citet{lee2019review}. Let $\mathcal{G} = (\mathcal{V}, \mathcal{E})$ be a random graph with vertex set $\mathcal{V} = [N]$ and where $\mathcal{E} \subseteq \mathcal{V} \times \mathcal{V}$ is a set of random edges. In a stochastic block model, each vertex belongs to one of $K$ communities called blocks and the probability of an edge forming between two vertices depends only on their block memberships. 

\begin{definition}\label{def:SBM}(Stochastic block model).
    Let $B \in \mathbb{R}^{K \times K}$ be a positive semi-definite matrix of block connection probabilities where $K \in \mathbb{N}$ is the number of blocks. Let $\pi = (\pi_{1},\dots,\pi_{N})^{\prime} \in \mathcal{S}_{K-1}$ represent the prior probabilities of each node belonging to the $k$-th community for $k \in [K]$. We say $A$ is the adjacency matrix for a stochastic block model, written $A\sim\mathrm{SBM}(B,\pi)$, if the block assignment map $z :[N] \to [K]$ satisfies $\mathbb{P}\{z(i) = k\} = \pi_{k}$ for each $k \in [K]$ and, conditional on $z$, the entries $A_{ij}$ are independently generated as 
    \begin{align}
        A_{ij} = \begin{cases}
            \text{Bernoulli}\left(B_{z(i),z(j)}\right) \quad &\text{for} \quad i \neq j \\
            1  \quad &\text{for} \quad i = j.
        \end{cases}
    \end{align}
\end{definition}

\begin{remark}
    The condition that $B$ is positive semi-definite in Definition \ref{def:SBM} corresponds to an assortative graph, referring to the tendency of similar vertices to be connected to each other \citep[][]{newman2003mixing}. In a network time series context, this amounts to assuming greater co-movement between series within the same community than between communities. 
\end{remark}

We now define the FNIRVAR model in which the common component is modelled as a (restricted) dynamic factor model and the idiosyncratic component is modelled as a restricted VAR process whereby the idiosyncratic VAR restrictions are determined by the adjacency matrix of a SBM. 

\begin{definition}\label{def:fnirvar}(FNIRVAR model).
    Let $\{ \boldsymbol{X}_t \}_{t \in \mathbb{Z}}$ denote a zero mean, second order stationary stochastic process where $\boldsymbol{X}_t = (X_{1,t}, \cdots, X_{N,t})' \in \mathbb{R}^N$. The FNIRVAR model is defined as:
    \begin{align}
    \boldsymbol{X}_t &= \Lambda \boldsymbol{F}_t + \boldsymbol{\xi}_t, \\
    \boldsymbol{F}_t &= \sum_{k=1}^{l_F} P_k \boldsymbol{F}_{t-k} + %N(0)
    \Gamma_u^{1/2}\boldsymbol{u}_t, \quad \boldsymbol{u}_t \sim \mathcal{N}(0, I_q), \\
    %I_q), \nonumber \\
    \xi_{i,t} &= \sum_{j=1}^{N} A_{ij} \tilde{\Phi}_{ij} \xi_{j,t-1} + \epsilon_{i,t}, \quad \boldsymbol{\epsilon}_t \sim \mathcal{N}(0, \Gamma_\epsilon), \\
    A &\sim \text{SBM}(B, \pi), \label{eq:model-def}
\end{align} 
where $\Lambda$ is a $N \times r$ matrix of loadings, $\boldsymbol{F}_t$ is an $r$-dimensional vector of factors, where $r\in\mathbb{N}$ is the number of factors, $P_k$ are $r \times r$ coefficient matrices of the factor VAR process with $l_F\in\mathbb{N}$ lags, $\Gamma_u$ is an $r \times q$ matrix, and $\boldsymbol{u}_t$ are common shocks. We further define $P \coloneqq (P_1; \ldots; P_{l_F}) \in \mathbb{R}^{rl_F \times r}$. The common component is written as $\bm{\chi}_{t} = \Lambda \boldsymbol{F}_{t}$, whereas the idiosyncratic component is $\boldsymbol{\xi}_{t}$. The parameters $\tilde{\Phi}_{ij}$ are taken as fixed. We also define $\Phi \coloneqq A \odot \tilde{\Phi}$. 
\end{definition}

% {\color{red} Expand a bit on the intuition behind this model, especially the role of $\Phi$ could be unclear here if the reader is not familiar with NIRVAR.} 

The NIRVAR idiosyncratic model can be interpreted as a restricted VAR model in which the restrictions on $\Phi$ are determined by the adjacency matrix of an assortative SBM. As such, $\Phi$ will have a block diagonal structure with a greater number of non-zero entries within each block than between blocks. The sparsity level of $\Phi$ is determined by the block probability matrix $B$. 

% \begin{remark}
    For ease of presentation, we consider the standard unweighted SBM and post multiply each realised edge $A_{ij}$ by a fixed weight $\tilde{\Phi}_{ij}$. Equivalently, one could model $\Phi$ as being the adjaceny matrix of a weighted SBM \citep[for a definition, see][]{gallagher2024spectral}. In this formulation, $\Phi_{ij} \mid z(i),z(j) %\overset{\text{ind}}{\sim} 
    \sim H\{z(i),z(j)\}$ for some family of distributions $H$ that depends only on the block assignments $z(i)$, and not on the individual vertices.
% \end{remark}

% \textbf{Definition 2.1.} \textit{Let $\{ \boldsymbol{X}_t \}_{t \in \mathbb{Z}}$ denote a zero mean, second order stationary stochastic process where $\boldsymbol{X}_t = (X_{1,t}, \cdots, X_{N,t})' \in \mathbb{R}^N$. The FNIRVAR model is}
% \begin{align}\label{eq:model-def}
%     \boldsymbol{X}_t &= \Lambda \boldsymbol{F}_t + \boldsymbol{\xi}_t \nonumber \\
%     \boldsymbol{F}_t &= \sum_{k=1}^{l_F} P_k \boldsymbol{F}_{t-k} + N(0)\boldsymbol{u}_t, \quad \boldsymbol{u}_t \sim \mathcal{N}(0, I_q) \nonumber \\
%     \xi_{i,t} &= \sum_{j=1}^{N} A_{ij} \tilde{\Phi}_{ij} \xi_{j,t-1} + \epsilon_{i,t}, \quad \boldsymbol{\epsilon}_t \sim \mathcal{N}(0, \Gamma_\epsilon), \nonumber \\
%     A &\sim \text{SBM}(B, \pi),
% \end{align}

% \noindent
% \textit{where $\Lambda$ is a $N \times r$ matrix of loadings, $\boldsymbol{F}_t$ is an $r$-dimensional vector of factors, $P_k$ are $r \times r$ coefficient matrices of the factor VAR process, $N(0)$ is an $r \times q$ matrix, and $\boldsymbol{u}_t$ are common shocks. We further define $\boldsymbol{P} := (P_1; \ldots; P_{l_F}) \in \mathbb{R}^{rl_F \times r}$. The idiosyncratic component follows a NIRVAR process with $\tilde{\Phi}$ being a fixed non-random matrix of weights. The block matrix of the stochastic block model (SBM) is $B \in [0,1]^{K \times K}$, $K \in \mathbb{Z}$. Let $Z$ be the $N \times K$ binary matrix whose entries satisfy $Z_{ik} = 1$ when vertex $i$ belongs to community $k \in [K]$ and zero otherwise. We also define $z : [N] \to [K]$ to be the block label of vertex $i$ and write $z_i \equiv z(i)$.}

It is worth pointing out that, in expectation, the idiosyncratic NIRVAR model can itself be viewed as a factor model. Let 
\begin{equation} 
\mathbb{E}\{\Phi\mid z(i) = k,z(j)=l\} = \tilde{B}_{kl},
\end{equation}
where $\tilde{B} \in \mathbb{R}^{K \times K}$ is the block mean matrix. Then the expected value of the observed weighted adjacency matrix can be written as $\mathbb{E}(\Phi) = Z\tilde{B}Z^{\prime}$, with $Z \in \{0,1\}^{N \times K}$ a binary matrix whose entries satisfy $Z_{ik} = 1$ when vertex $i$ belongs to community $k \in [K]$ and zero otherwise. 
% If we replace $\Phi$ with $\mathbb{E}(\Phi)$, the NIRVAR model becomes 
Replacing $\Phi$ with $\mathbb{E}(\Phi)$, the NIRVAR model yields
\begin{align}
    \bm{\xi}_{t} = Z\tilde{B}Z^{\prime} \bm{\xi}_{t-1} + \boldsymbol{\epsilon}_t = \Lambda_{\xi} \bm{F}^{(\xi)}_{t} + \boldsymbol{\epsilon}_{t},
\end{align}
where $\bm{F}^{(\xi)}_{t} \coloneqq Z^{\prime} \bm{\xi}_{t-1}$ is a $K$-dimensional vector of factors encoding the aggregated response in each community and $\Lambda_{\xi} \coloneqq Z\tilde{B}$ are the loadings of each panel component onto these ``community factors''. This relation between the expected adjacency matrix and community factors is also discussed in \citet[][]{chen2023community}.

Before describing the FNIRVAR estimation procedure, we analyse the conditions under which the model is stationary.
 
\begin{proposition}[Stationarity]
\label{prop:stability}
The FNIRVAR process is weakly stationary if and only if 
% \begin{equation}
$\det(I_{Nl_F} - P^{*} z) \neq 0$ and $ \det(I_N - \Phi z) \neq 0$ for $|z| \leq 1$,
% \end{equation}
where
\begin{equation}
P^* = 
\begin{bmatrix}
P_1 & P_2 & P_3 & \cdots & P_{l_F} \\
I_N & 0   & 0   & \cdots & 0 \\
0   & I_N & 0   & \cdots & 0 \\
\vdots & \ddots & \ddots & \ddots & \vdots \\
0 & \cdots & 0 & I_N & 0
\end{bmatrix}
\end{equation}
is the companion matrix of the factor VAR process in~\eqref{eq:model-def}.
\end{proposition}

\begin{proof}
    The VAR processes $\{\bm{F}_{t}\}_{t \in \mathbb{Z}}$ and $\{\bm{\xi}_{t}\}_{t \in \mathbb{Z}}$ are weakly stationary if and only if $\text{det}(I_{r} - P_{1}z - \cdots - P_{l_{F}}z^{l_{F}}) \neq 0$ and $\text{det}(I_{N} - \Phi z) \neq 0$ for all $z \in \mathbb{C}$ such that $|z| \leq 1$ \citep[see][for example]{brockwell1991time}. By \citet[][Lemma 2.1]{tsay2013multivariate}, $\text{det}(I_{r} - P_{1}z - \cdots - P_{l_{F}}z^{l_{F}}) = \det(I_{Nl_F} - P^{*} z)$. By \citet[][Theorem 1]{lutkepohl1984linear}, if $\Lambda \neq 0$, then $\{\Lambda \bm{F}_{t}\}_{t \in \mathbb{Z}}$ is weakly stationary whenever $\det(I_{Nl_F} - P^{*} z) \neq 0$ for $|z| \leq 1$. Finally, by \citet[][]{granger1976time}, the sum of the processes $\{\Lambda \bm{F}_{t}\}_{t \in \mathbb{Z}}$ and $\{\bm{\xi}_{t}\}_{t \in \mathbb{Z}}$ is a stationary vector autoregressive moving average process whenever the individual processes are stationary. 
\end{proof}

%% file: Sections/estimation.tex
\section{Estimation}\label{sec:estimation}

We adopt a two-step estimator of the FNIRVAR parameters $\Lambda$, $\boldsymbol{F}_{t}$, $P_{k}$, and $\Phi$. In the first step, the factors {$\boldsymbol{F}_t$} and loadings {$\Lambda$} of the factor model are estimated via PCA. After removing the estimated common component, the latent VAR parameters {$\Phi$} of the idiosyncratic model are then estimated in a second step. The following assumption of a large eigengap is required for this %two-step estimation 
procedure.

\begin{assumption}
\label{ass:strong}
There exists a positive eigengap \(\Delta := \lambda_{\min}(\Gamma_{\chi}) - \lambda_{\max}(\Gamma_\xi)\) between the smallest nonzero eigenvalue of the common covariance matrix $\Gamma_{\chi}$ and the largest eigenvalue of the idiosyncratic covariance matrix $\Gamma_{\xi}$.
\end{assumption}

Consider a set of observations \(X = (\boldsymbol{x}_1, \ldots, \boldsymbol{x}_T)\) where \(\boldsymbol{x}_t = (x_{1t}, \ldots, x_{Nt})'\) is a realisation of the random variable \(\boldsymbol{X}_t\). Let \(E := \text{eig}(XX'/T)_{[:,1:r]}\) be the \(N \times r\) matrix of eigenvectors of the sample covariance matrix corresponding to the \(r\) largest eigenvalues. We estimate the factors and loadings via static principal component analysis as $\widehat{\Lambda} = E$ and $\widehat{F}_t = E' \boldsymbol{x}_t$. For prediction, the next step ahead factor estimate is given by
\[
\widehat{F}_{t+1} = \sum_{k=1}^{\ell_F} \widehat{P}_k \widehat{F}_{t-k},
\]
where \(P\) is estimated by least-squares regression. In particular, if \(Y_F = (F'_{\ell_F}, \ldots, F'_T)' \in \mathbb{R}^{(T - \ell_F) \times r}\) and
\[
X_F = 
\begin{pmatrix}
F'_0 & \cdots & F'_{\ell_F - 1} \\
\vdots & \ddots & \vdots \\
F'_{T - \ell_F - 1} & \cdots & F'_{T - 1}
\end{pmatrix}
\in \mathbb{R}^{(T - \ell_F) \times r \ell_F},
\]
then
\begin{equation}
\widehat{P} = (\hat{X}_{F}^{\prime} {\hat{X}_{F}})^{-1} \hat{X}_{F}^{\prime} \hat{Y}_{F}.
\end{equation}

We choose the number of factors $r$ using the information criteria of \citet{bai2002determining} and the order $l_{F}$ using AIC \citep{fan2023bridging}. %, as suggested by \citet{fan2023bridging}. 
Once the common component has been estimated as $\hat{\bm{\chi}}_{t} = \hat{\Lambda}\hat{\bm{F}}_{t}$, we estimate the idiosyncratic component as $\hat{\bm{\xi}}_{t} = \bm{X}_{t} - \hat{\bm{\chi}}_{t}$. Let $\hat{\chi} = (\hat{\bm{\chi}}_{1},\dots,\hat{\bm{\chi}}_{T})$ and $\hat{\xi} = (\hat{\bm{\xi}}_{1},\dots,\hat{\bm{\xi}}_{T})$. The parameters of the NIRVAR model for the idiosyncratic component are then estimated via Algorithm \ref{alg:NIRVAR}. Details on the final step of Algorithm~\ref{alg:NIRVAR} of estimating a restricted VAR model via OLS can be found in Chapter 5 of \citet{lutkepohl2005new}. 
% \begin{enumerate}
%     \item Compute the sample covariance matrix $S_{\xi} = \hat{\xi}\hat{\xi}^{\prime}/T$. 
%     \item Compute $d$ eigenvectors $\bm{u}_{1},\dots,\bm{u}_{d}$ associated with the\\largest (in magnitude) eigenvalues of $S_{\xi}$ and form the $N \times d$ matrix $U$ with columns $\bm{u}_{1},\dots,\bm{u}_{d}$. 
%     \item Cluster the rows of $U$ via a Gaussian mixture model with $K$ components and use the cluster assignments of the rows as estimates for $z(i)$ for each $i \in [N]$. 
%     \item Define $\hat{A}_{ij} = \indfun\{\hat{z}(i) = \hat{z}(j)\}$.
%     \item Estimate the parameters $\Tilde{\Phi}$ of the subset VAR model $\bm{\xi}_{t} = \hat{A}\odot \tilde{\Phi} \bm{\xi}_{t-1} + \bm{\epsilon}^{(\hat{A})}_{t}$ via ordinary least squares \citep[see Chapter 5 of][for details on estimation of restricted VAR models]{lutkepohl2005new}. 
% \end{enumerate}

\begin{algorithm}[t]
\caption{FNIRVAR estimation.}
\label{alg:NIRVAR}
\begin{algorithmic}[1]  
  \Require $N \times T$ design matrix $X$; number of factors $r$; factor VAR lag order $l_{F}$; embedding dimension $d$; number of clusters $K$.
  \State Estimate $r$ factors and corresponding loadings via static PCA: $\widehat{\Lambda} = E$ and $\widehat{F}_t = E' \boldsymbol{x}_t$ where $E := \text{eig}(XX'/T)_{[:,1:r]}$.
  \State Estimate $l_{F}$ factor VAR coefficient matrices $P_{k}$ via least-squares.
  \State Estimate the idiosyncratic component as $\hat{\bm{\xi}}_{t} = \bm{X}_{t} - \hat{\bm{\chi}}_{t}$.
  \State Compute the sample covariance matrix
        $\hat{\Gamma}_{\xi} \gets \hat{\xi}\hat{\xi}^{\prime}\!/\!T$.
  \State Obtain the $d$ eigenvectors $\bm{v}_1,\dots,\bm{v}_d$ corresponding to the $d$ largest eigenvalues in magnitude of $\hat{\Gamma}_{\xi}$, and form the $N \times d$ matrix $V$ with columns $\bm{v}_{1},\dots,\bm{v}_{d}$.
  \State Cluster the rows of $V$ via a $K$-component Gaussian mixture model. 
  \State Set $\hat{z}(i)$ to the cluster assignment of row $i$ for each $i\in[N]$.
  \State Define $\hat{A}_{ij} = \indfun\{\hat{z}(i) = \hat{z}(j)\}$.
  \State Estimate the parameters $\Tilde{\Phi}$ of the restricted VAR model $\bm{\xi}_{t} = \hat{A}\odot \tilde{\Phi} \bm{\xi}_{t-1} + \bm{\epsilon}^{(\hat{A})}_{t}$ via OLS.
  
  \Returnx Factors $\widehat{F}_t$, loadings $\widehat{\Lambda}$, factor VAR coefficient matrices $\widehat{P}_{k},$ cluster labels $\hat{z}(i)$; estimated coefficient matrix $\hat{\Phi}$. 
\end{algorithmic}
\end{algorithm}

Some remarks are in order motivating the NIRVAR estimation procedure in Algorithm~\ref{alg:NIRVAR}. %It is well known that, 
Under suitable sparsity and noise levels, spectral clustering of the adjacency matrix of an SBM consistently recovers the blocks of the SBM \citep[][]{lei2015consistency}. In our case, the adjacency matrix is unobserved and we instead use the sample covariance matrix $\hat{\Gamma}_{\xi}$ as the similarity matrix to embed. This is reasonable since the covariance matrix $\Gamma_{\xi}$ inherits the block structure of the adjacency matrix through the relation $\text{vec}(\Gamma_{\xi}) = (I_{N^{2}} - \Phi \otimes \Phi)^{-1}\text{vec}(\Gamma_{\epsilon})$ \citep[][]{lutkepohl2005new}. As proven by \citet{martin2024nirvar}, the hard thresholding given by step (4) above yields a biased estimator whenever the network contains edges between nodes in different communities. On the other hand, the advantages of the hard thresholding are a reduction in variance of the VAR estimates, as well as regularisation which allows for estimation in a high-dimensional setting.

Various methods for choosing the embedding dimension $d$ have been proposed in the literature, such as the profile likelihood method of \citet{zhu2006automatic} or the ScreeNOT method of \citet{donoho2023screenot}. We choose $d$ using tools from random matrix theory that are standard in finance and economics \citep[see][for example]{laloux2000random}. In particular, $d$ is set equal to the number of eigenvalues of $\hat{\Gamma}_{\xi}$ that are larger than the upper bound of the support of the Mar\v{c}enko-Pastur distribution \citep[][]{marvcenko1967distribution}. Since the Mar\v{c}enko-Pastur distribution is the limiting eigenvalue distribution of a Wishart matrix, the interpretation is that those eigenvalues that are greater than the upper bound of its support cannot be attributed to random noise and are thus ``informative directions''.

The number of clusters $K$ could be chosen using an information criterion such as BIC. However, we follow the argument of \citet{martin2024nirvar} and set $K=d$ since, when represented as a latent position random graph \citep[][]{hoff2002latent}, the expected adjacency matrix of an SBM has $K$ distinct latent positions each of dimension $d=K$.

The computational complexity of FNIRVAR is dominated by estimation of the idiosyncratic component, which is $O(N^{3}/K^{2} + TN^{2}/K)$ when $\Gamma_{\epsilon}$ is homoskedastic and $O(N^{3} + TN^{2})$ otherwise (since generalised least-squares estimation must be used in this case) \citep[][]{martin2024nirvar}. Estimation of the factors and loadings via PCA on the other hand can be done using a Lanczos type algorithm in $O(TNr)$ time. Least-squares estimation of the factor VAR process is $O(Tr^{2}l_{F}^{2})$ and is therefore subdominant since $r \ll N$ and $l_{F} \ll N$. We note that least-squares estimation of an unrestricted VAR is $O(N^{3}+ TN^{2})$, and thus, FNIRVAR estimation is faster whenever $\Gamma_{\epsilon}$ is assumed to be homoskedastic.

%% file: Sections/simulation1.tex
\section{Simulation study}\label{sec:sim}

% {\color{red} Add a few lines introducing the simulations in this section.} 

We conduct three simulation studies investigating the relationship between the common and idiosyncratic components of FNIRVAR. Firstly, we generate data from the FNIRVAR model and conduct a backtesting experiment comparing the predictive performance of a factor model with no idiosyncratic model, a factor model with a LASSO estimated sparse VAR model, and the FNIRVAR prediction model outlined in Section \ref{sec:estimation}. As expected, given the data generating process, FNIRVAR achieves the lowest mean squared prediction error (MSPE). 
Secondly, we look at the predictive performance of FNIRVAR when the generative model is a static factor model. The study highlights the benefits of including the FNIRVAR idiosyncratic model when the estimated number of factors is misspecified. Thirdly, we consider the effect of Assumption \ref{ass:strong} on the sample covariance eigenvalues under a FNIRVAR generative model. In particular, the study highlights the necessity of Assumption \ref{ass:strong} for the proposed two-step estimation procedure. 

\subsection{Prediction under %an 
FNIRVAR} %data generating process} 

We generate data from an FNIRVAR model with $N=100$, $T=1500$, $r=5$, $q=5$, $l_{F} = 2$, $K=4$, $B_{ij} = 0.9$ for $i=j$ and $B_{ij} = 0.1$ for $i \neq j$. $P_{k}$ was chosen to have unit diagonal entries and off diagonal entries equal to -0.2, with a subsequent rescaling such that $\rho(P^{*})=0.7$. The weights for the NIRVAR coefficient matrix are chosen as $\Tilde{\Phi}_{ij} \sim \mathcal{N}(\mu_{i},1)$ where $\mu_{i} \coloneqq (-1)^{z(i)+1}z(i)$, and $\Phi = A \odot \Tilde{\Phi}$ is then normalised such that $\rho(\Phi) = 0.9$. The error covariance matrices $\Gamma_{u}$ and $\Gamma_{\epsilon}$ are both set to be the identity. The task is one-step ahead prediction using a rolling window backtest framework with a lookback window of 1000 observations. Table \ref{tab:simFNIRVAR} shows the MSPE and corresponding standard error of the MSPE for three different prediction models: FNIRVAR, a static factor model with no idiosyncratic model, and a static factor model with a LASSO estimated sparse VAR model. The LASSO penalty was chosen via BIC. FNIRVAR achieves the lowest MSPE as expected given that the data generating process is the FNIRVAR model. 

\begin{table}[t]
\scalebox{0.9}{
\begin{tabular}{cccc}
\toprule
 & FNIRVAR  & Factors Only & Factors + LASSO  \\
\midrule
MSPE & \textbf{1.87}  & 1.91 & 1.95 \\
SE ($\times 10^{-1})$ & \textbf{1.18} & 1.21 & 1.23 \\
\bottomrule
\end{tabular}
}
\caption{MSPE and corresponding standard error (SE) for three prediction models when the DGP is FNIRVAR.} \label{tab:simFNIRVAR}%
\end{table}

\subsection{Misspecifying the number of factors}
We generate data from a static factor model with $N=100$, $T=1500$, $r=5$ and $l_{F} = 2$. The number of common shocks was $q=5$ and $P_{k}$ was chosen to have unit diagonal entries and off diagonal entries equal to -0.2, with a subsequent rescaling such that $\rho(P^{*})=0.7$. The task is one-step ahead prediction using a rolling window backtest framework with a lookback window of 1000 observations. Table \ref{table:sim-misspecify} shows the MSPE and corresponding standard error for static factor and FNIRVAR prediction models, assuming different $r$ values for the predictions. For FNIRVAR, the number of blocks was set as $K=4$. Even though the data generating process has no idiosyncratic component, FNIRVAR outperforms a static factor model when the number of factors is misspecified. The NIRVAR idiosyncratic model captures some of the variance left over from under-specifying $r$. Nonetheless, FNIRVAR does not adversely affect the predictive performance when $r$ is correctly specified. 

\begin{table}[t]
\centering
\scalebox{0.9}{
\begin{tabular}{c|ccc|ccc}
    \toprule
    Metric
    & \multicolumn{3}{c|}{MSPE} 
    & \multicolumn{3}{c}{Standard error }   \\
    \midrule
    \diagbox{Model}{$\hat{r}$}
    & 1  & 3 & 5 & 1 & 3 & 5  \\
    \midrule
    FNIRVAR & 1.41  & 1.29 & \textbf{1.22}  & 0.09  & 0.08 & \textbf{0.07}  \\
    Factors Only  & 3.20 & 1.51  & \textbf{1.22}  & 0.22  & 0.32  & \textbf{0.07}    \\
    \bottomrule
\end{tabular}
}
\caption{MSPE and corresponding standard error of two different prediction models for a backtesting experiment on data generated by a static factor model with $r = 5$.}
\label{table:sim-misspecify}
\end{table}

\subsection{Violating the positive eigengap assumption}

We simulate data from the FNIRVAR model for different values of $N$, with $T=1000$, and monitor the growth of the eigenvalues of $\hat{\Gamma}$ as $N$ increases. For this simulation, we choose a single factor ($r=1$) and five blocks ($K = 5$). The loadings $\Lambda$ are sampled from a mixture of normals: $\Lambda_{i,j} \sim 0.5\cdot \mathcal{N}(1,\sigma_{\Lambda}^{2}) + 0.5\cdot\mathcal{N}(-1,\sigma_{\Lambda}^{2})$. The factor VAR coefficient matrices are sampled independently from $\mathcal{N}(0,0.1)$ and then rescaled such that $\rho(P^{*}) = 0.6$ to ensure stationarity.  The number of common shocks is chosen to be equal to the number of factors ($q=r$) and $\Gamma_u$ is set to the identity matrix. The number of factor lags is chosen to be $l_{F} = 5$. The block matrix $B$ is set to have diagonal values of $0.95$ and off-diagonal values of $0.01$. The weights for the NIRVAR coefficient matrix are chosen as $\Tilde{\Phi}_{ij} \sim \mathcal{N}(\mu_{i},1)$ where $\mu_{i} \coloneqq (-1)^{z(i)+1}z(i)$, then symmetrized as $\Tilde{\Phi} \leftarrow (\Tilde{\Phi} + \Tilde{\Phi}^{\prime})/2$. The NIRVAR coefficient matrix $\Phi = A \odot \Tilde{\Phi}$ is then normalised such that $\rho(\Phi) = 0.9$. The idiosyncratic error covariance matrix is chosen to be the identity: $\Gamma_{\epsilon} = I_{N}$. For each $N$, 20 replicate datasets are simulated using the above hyperparameters. We control the eigengap $\Delta$ through the choice of $\sigma_{\Lambda}$. For $\sigma_{\Lambda}^{2} = 9 \times10^{-3}$, $\Delta > 0$ whereas for $\sigma_{\Lambda}^{2} = 2.5 \times 10^{-4}$, $\Delta < 0$. Figure \ref{fig:sample-evals} shows the minimum eigenvalue of $\hat{\Gamma}_{\chi}$ and maximum eigenvalue of $\hat{\Gamma}_{\xi}$ as a function of $N$ for the two scenarios: $\Delta > 0$ and $\Delta < 0$. In both cases, the eigenvalue of $\hat{\Gamma}_{\chi}$ grows linearly in $N$, as is expected for a strong factor \citep[see][for further discussion on the growth rate of eigenvalues of factor models]{barigozzi2024dynamic}. However, when $\Delta <0$ (the case where Assumption \ref{ass:strong} is violated), $\lambda_{\text{max}}(\hat{\Gamma}_{\chi}) < \lambda_{\text{min}}(\hat{\Gamma}_{\xi}) $, and therefore, removing the top eigenvalues of $\hat{\Gamma}$, as in our two-step estimation procedure, will not pick out the common component. As a further observation, we note that the eigenvalues of $\hat{\Gamma}$ do not grow linearly with $N$, confirming that the idiosyncratic component of FNIRVAR cannot be interpreted as a strong factor model.

\begin{figure}[t]
\captionsetup{skip=0pt}
\centering
\begin{subfigure}[t]{.235\textwidth}
  \centering
  \caption{$\Delta>0$}
  \includegraphics[width=.99\linewidth]{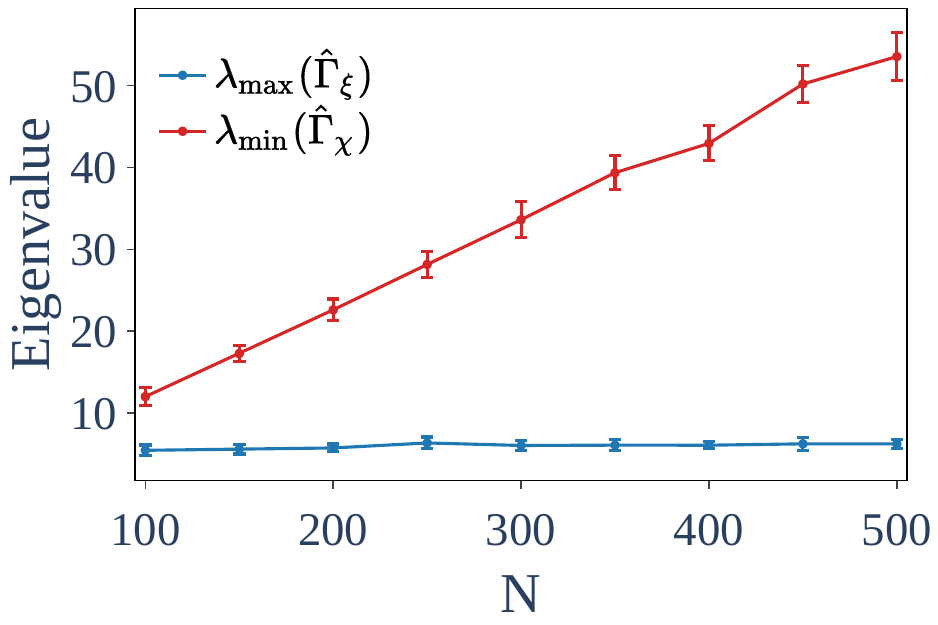}
  \label{fig:fred-gt-sub1}
\end{subfigure} 
\hfill
\begin{subfigure}[t]{.235\textwidth}
  \centering
  \caption{$\Delta <0$}
  \includegraphics[width=.99\linewidth]{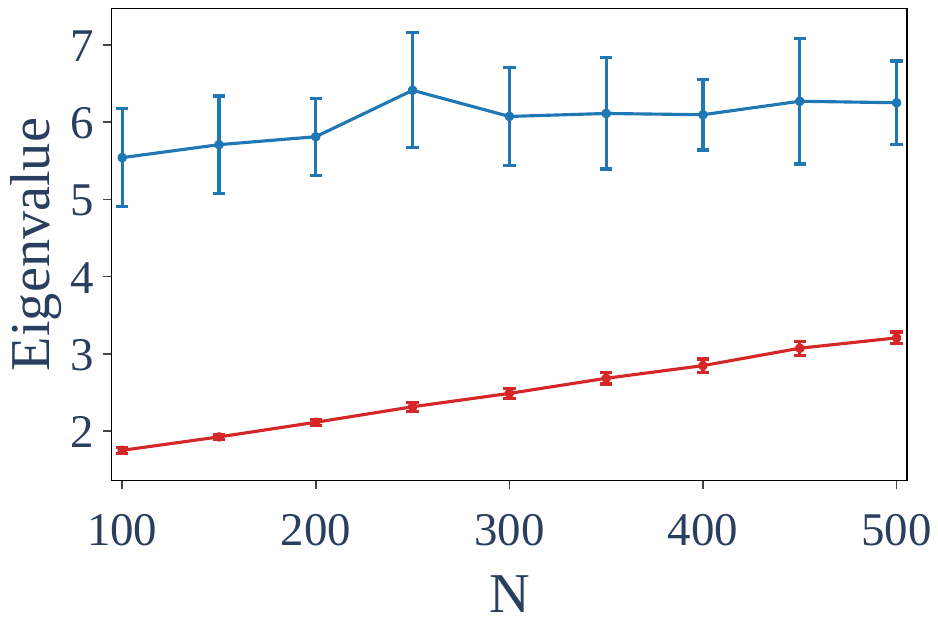}
  \label{fig:fred-gt-sub2}
\end{subfigure}
%\vspace{-4mm}
\caption{Growth of eigenvalues of $\hat{\Gamma}$ with $N$ under a FNIRVAR model when Assumption \ref{ass:strong} (a) holds and (b) is violated. }
\label{fig:sample-evals}
\end{figure}

%% file: Sections/daily_returns.tex
\subsection{Application to daily returns prediction}\label{subsec:daily}

\paragraph{Data description}
The previous close-to-close (pvCLCL) price returns of $N = 648$ financial assets between 03/01/2000 and 31/12/2020 ($T = 5279$) were derived from databases provided by the Center for Research in Security Prices, LLC, an affiliate of the University of Chicago Booth School of Business. The assets considered include stocks in the S\&P 500 and Nasdaq as well as a small number of exchange-traded funds. Using the S\&P 500 index as a proxy for the market, we construct the pvCLCL market excess returns by subtracting the return of SPY, the exchange traded fund which tracks the S\&P 500 index. 

\paragraph{Task}
The task is to predict the next day pvCLCL market excess returns. We adopt a rolling window backtesting framework with a fixed four year look-back. The predictive performance between 01/01/2004 and 31/12/2020 of the three models are compared via commonly used metrics in the financial literature \citep{gu2020empirical}. 
The mean number of factors chosen by %the 
PCp2 %criterion 
\citep{bai2002determining} over the backtest period was $ 12.16$, whilst the mean order of the factor VAR model was $ 3.97$.

\paragraph{Performance metrics}
Each day we measure the profit and loss as $\text{PnL}_{t} = \sum_{i=1}^{N} \omega_{i}^{(t)} \text{sign}( \hat{s}_{i}^{(t)} )\, s_{i}^{(t)}$ where $ \hat{s}_{i}^{(t)} $ is the predicted return of asset $i$ on day $t$, $s_{i}^{(t)}$ is the realised return of asset $i$ on day $t$, and  $\omega_{i}^{(t)}$ are  portfolio weights such that $\sum_{i=1}^{N} \omega_{i}^{(t)}=1$ for each $t \in [T]$. Letting $\text{PnL} \coloneqq \{\text{PnL}_{t}\}_{t = 1,\dots,T}$, the annualised Sharpe Ratio (SR) is then given by $\sqrt{252} \times \text{mean}(\text{PnL})/\text{stdev}(\text{PnL})$. Furthermore, we compute the mean daily PnL (given in basis points, denoted \textit{bpts}, with 1\% = 100bpts).
% as well as the mean daily percentage of correct predictions, called the hit ratio.
Fixed transaction costs of 1bpts and 2bpts are also considered, whereby a fixed amount is subtracted from the PnL every time we flip our position in a
given asset from long (short) to short (long). This is a realistic level of transaction cost that one might expect from a broker.

\paragraph{Equal-weighted portfolio} 
The simplest possible porfolio construction strategy is an equal weighted portolio defined by uniform weights, $\omega_{i}^{(t)} = 1/N$ for all $i \in [N], t \in[T]$. Table \ref{table:daily-horserace} gives the Sharpe ratios and mean daily PnL for each model for different levels of fixed transaction costs. FNIRVAR has the highest Sharpe ratios and mean daily PnL values across all transaction cost levels, followed by the LASSO estimated factor plus sparse VAR model. This suggests that modelling the idiosyncratic component on top of the common component can lead to significant economic gains.

We also consider decile portfolios, in which solely the top $x$\% largest in magnitude signals are traded each day \citep{fama1992cross}. The Sharpe ratio and mean daily PnL (using FNIRVAR as the prediction model) for equal weighted decile portfolios with different transaction costs are shown in Table \ref{tab:equal_decile_transaction_combined}. 
% The table shows that there is a trade-off between Sharpe ratio and PnL: lower decile portfolios have higher PnL but lower Sharpe ratios. 
The table shows that lower decile portfolio have higher PnL, with the Sharpe ratios maintaining a similar level.

\begin{table}[t]
\centering
\scalebox{0.8}{
\begin{tabular}{c|ccc|ccc}
    \toprule
    Metric
    & \multicolumn{3}{c|}{Sharpe ratio} 
    & \multicolumn{3}{c}{Mean daily PnL }   \\
    \midrule
    \diagbox{Model}{Cost}
    & 0 bpts  & 1 bpts & 2 bpts & 0 bpts & 1 bpts & 2 bpts  \\
    \midrule
    FNIRVAR & \textbf{1.95}  & 1.66 & 1.37  & \textbf{3.41}  & 2.90 & 2.40  \\
    Factors + LASSO  & 1.64  & 1.29 & 0.94 & 2.36 & 1.86 & 1.35  \\
    Factors Only  & 1.27 & 1.07 & 0.86 & 3.09 & 2.59 & 2.09   \\
    \bottomrule
\end{tabular}
}
\caption{Sharpe ratio, mean daily PnL (in bpts) for three prediction models under an equal weighted portfolio construction, for different levels of fixed transaction costs.}
\label{table:daily-horserace}
\end{table}

% \begin{table}[htbp]
% \centering
% \scalebox{0.8}{
% \begin{tabular}{c|ccc|ccc}
%     \toprule
%     Metric
%     & \multicolumn{3}{c|}{Sharpe ratio} 
%     & \multicolumn{3}{c}{Mean daily PnL} \\
%     \midrule
%     \diagbox{Decile}{Cost}
%     & 0 bpts & 1 bpt & 2 bpts & 0 bpts & 1 bpt & 2 bpts \\
%     \midrule
%     100 & 1.80 & 1.54 & 1.28 & 3.41 & 2.91 & 2.41 \\
%     75  & 1.78 & 1.57 & 1.36 & 4.23 & 3.73 & 3.23 \\
%     50  & 1.74 & 1.57 & 1.40 & 5.21 & 4.71 & 4.21 \\
%     25  & 1.56 & 1.43 & 1.31 & 6.19 & 5.69 & 5.19 \\
%     \bottomrule
% \end{tabular}
% }
% \caption{Sharpe ratio and mean daily PnL (in basis points) for equal weighted decile portfolios under different transaction costs.}
% \label{tab:equal_decile_transaction_combined}
% \end{table}

\begin{table}[t]
\centering
\scalebox{0.8}{
\begin{tabular}{c|ccc|ccc}
    \toprule
    Metric
    & \multicolumn{3}{c|}{Sharpe ratio} 
    & \multicolumn{3}{c}{Mean daily PnL} \\
    \midrule
    \diagbox{Decile}{Cost}
    & 0 bpts & 1 bpt & 2 bpts & 0 bpts & 1 bpt & 2 bpts \\
    \midrule
    100 & 1.95 & 1.66 & 1.37 & 3.41 & 2.90 & 2.40 \\
    75  & 1.99 & 1.76 & 1.53 & 4.35 & 3.85 & 3.34 \\
    50  & \textbf{2.02} & 1.83 & 1.65 & 5.50 & 4.99 & 4.49 \\
    25  & 1.90 & 1.76 & 1.62 & \textbf{6.81} & 6.30 & 5.80 \\
    \bottomrule
\end{tabular}
}
\caption{Sharpe ratio and mean daily PnL (in bpts) for equal weighted decile portfolios under different transaction costs.}
\label{tab:equal_decile_transaction_combined}
\end{table}

\paragraph{Value-weighted portfolio}
A value-weighted portfolio accounts for the liquidity in each asset and is defined by
%\begin{align}
$
    \omega_{i}^{(t)} =\min (\alpha \nu_{i,t} , \beta), 
$
%\end{align}
where 
%\begin{align}
$
    \nu_{i,t} = \text{median}\{V_{i,t},\dots,V_{i,1} \},
$
%\end{align}
with $V_{i,t}$ being the dollar volume traded for asset $i$ on day $t$. The parameters $\alpha$ and $\beta$ control the imbalance between portfolio weights: as $\alpha $ increases and $\beta$ decreases we get closer to the equal weighted portfolio setting. Table \ref{tab:alpha_beta_combined} shows that as $\alpha $ increases and $\beta$ decreases, the Sharpe ratio increases and portfolio variance decreases, respectively. The mean portfolio weight for every choice of $(\alpha,\beta)$ was 0.00154. We also show the total book value and mean bet size in dollars for different values of $(\alpha,\beta)$ in Table \ref{tab:alpha_beta_combined_betsize}. We note that the largest book size considered is $\$156.9$ million which is much smaller that the total market. Therefore, market impact will be negligible. For fixed $\alpha = 0.001$ and $\beta = 500,000$, Table \ref{tab:value_decile_transaction_combined} shows the Sharpe ratio and mean daily PnL for various decile portfolios and different transaction costs. Note that the choice of $(\alpha,\beta)$ in Table \ref{tab:value_decile_transaction_combined} is the furthest %\MC{farthest?} 
from an equal-weighted portfolio out of all pairs of $(\alpha,\beta)$ considered. Again, lower decile portfolios give larger PnL. 

% \begin{table}[htbp]
% \centering
% \scalebox{0.8}{
% \begin{tabular}{c|ccc|ccc}
%     \toprule
%     Metric
%     & \multicolumn{3}{c|}{Sharpe ratio}
%     & \multicolumn{3}{c}{Mean std.\ of weights} \\
%     \midrule
%     \diagbox{$\alpha$}{$\beta$}
%     & 100{,}000 & 200{,}000 & 500{,}000 & 100{,}000 & 200{,}000 & 500{,}000 \\
%     \midrule
%     0.001 & 1.01 & 0.89 & 0.82 & 0.00165 & 0.00203 & 0.00265 \\
%     0.005 & 1.23 & 1.14 & 1.01 & 0.00101 & 0.00124 & 0.00165 \\
%     0.010 & 1.31 & 1.23 & 1.11 & 0.00082 & 0.00101 & 0.00133 \\
%     \bottomrule
% \end{tabular}
% }
% \caption{Sharpe ratio and mean daily standard deviation of portfolio weights for different values of $\alpha$ and $\beta$.}
% \label{tab:alpha_beta_combined}
% \end{table}

\begin{table}[t]
\centering
\scalebox{0.8}{
\begin{tabular}{c|ccc|ccc}
    \toprule
    Metric
    & \multicolumn{3}{c|}{Sharpe ratio}
    & \multicolumn{3}{c}{Mean std.\ of weights} \\
    \midrule
    \diagbox{$\alpha$}{$\beta$}
    & 100{,}000 & 200{,}000 & 500{,}000 & 100{,}000 & 200{,}000 & 500{,}000 \\
    \midrule
    0.001 & 1.14 & 1.05 & 1.00 & 0.00165 & 0.00203 & 0.00265 \\
    0.005 & 1.34 & 1.26 & 1.14 & 0.00101 & 0.00124 & 0.00165 \\
    0.010 & \textbf{1.44} & 1.34 & 1.24 & 0.00082 & 0.00101 & 0.00133 \\
    \bottomrule
\end{tabular}
}
\caption{Sharpe ratio and mean daily standard deviation of portfolio weights for different values of $\alpha$ and $\beta$.}
\label{tab:alpha_beta_combined}
\end{table}

\begin{table}[t]
\centering
\scalebox{0.8}{
\begin{tabular}{c|ccc|ccc}
    \toprule
    Metric
    & \multicolumn{3}{c|}{Total bet size ($\times 10^{6}$)}
    & \multicolumn{3}{c}{Mean bet size per asset ($\times 10^{3}$)} \\
    \midrule
    \diagbox{$\alpha$}{$\beta$}
    & 100{,}000 & 200{,}000 & 500{,}000 & 100{,}000 & 200{,}000 & 500{,}000 \\
    \midrule
    0.001 & 23.7  & 32.2 & 41.5 & 36.6 & 49.7 & 64.1 \\
    0.005 & 40.8 & 67.7 & 118.5 & 63.0 & 104.5 & 182.9 \\
    0.010 & 46.8 & 81.6 & 156.9 & 72.3 & 126.0 & 242.1 \\
    \bottomrule
\end{tabular}
}
\caption{Bet sizes per asset for different values of $\alpha$ and $\beta$.}
\label{tab:alpha_beta_combined_betsize}
\end{table}

% \begin{table}[htbp]
% \centering
% \scalebox{0.8}{
% \begin{tabular}{c|ccc|ccc}
%     \toprule
%     Metric
%     & \multicolumn{3}{c|}{Sharpe Ratio}
%     & \multicolumn{3}{c}{Mean Daily PnL (bpts)} \\
%     \midrule
%     \diagbox{Decile}{Cost}
%     & 0 bpts & 1 bpt & 2 bpts & 0 bpts & 1 bpt & 2 bpts \\
%     \midrule
%     100 & 0.82 & 0.58 & 0.34 & 1.69 & 1.19 & 0.69 \\
%     75  & 0.82 & 0.64 & 0.45 & 2.20 & 1.68 & 1.18 \\
%     50  & 0.89 & 0.75 & 0.61 & 3.10 & 2.58 & 2.08 \\
%     25  & 0.78 & 0.68 & 0.58 & 3.87 & 3.37 & 2.87 \\
%     \bottomrule
% \end{tabular}
% }
% \caption{Sharpe ratio and mean daily PnL (in basis points) for value weighted decile portfolios under different transaction costs ($\alpha = 0.001$, $\beta = 500{,}000$).}
% \label{tab:value_decile_transaction_combined}
% \end{table}

\begin{table}[t]
\centering
\scalebox{0.8}{
\begin{tabular}{c|ccc|ccc}
    \toprule
    Metric
    & \multicolumn{3}{c|}{Sharpe Ratio}
    & \multicolumn{3}{c}{Mean Daily PnL (bpts)} \\
    \midrule
    \diagbox{Decile}{Cost}
    & 0 bpts & 1 bpt & 2 bpts & 0 bpts & 1 bpt & 2 bpts \\
    \midrule
    100 & 1.00 & 0.74 & 0.47 & 1.88 & 1.38 & 0.87 \\
    75  & \textbf{1.07} & 0.86 & 0.66 & 2.60 & 2.09 & 1.58 \\
    50  & 1.05 & 0.89 & 0.73 & 3.30 & 2.78 & 2.27 \\
    25  & 0.87 & 0.76 & 0.65 & \textbf{3.84} & 3.34 & 2.83 \\
    \bottomrule
\end{tabular}
}
\caption{Sharpe ratio and mean daily PnL (in basis points) for value weighted decile portfolios under different transaction costs ($\alpha = 0.001$, $\beta = 500{,}000$).}
\label{tab:value_decile_transaction_combined}
\end{table}

%% file: Sections/intraday_returns.tex
\subsection{Intraday returns prediction}\label{subsec:intraday}

\paragraph{Data description}
The minutely best-bid and ask prices for a universe of $N=516$ assets in the Nasdaq between 27/06/2007 and 31/12/2021 ($T=3656$) were extracted from the LOBSTER database\footnote{LOBSTER: Limit Order Book System - The Efficient Reconstructor at Humboldt Universität zu Berlin, Germany. \url{https://lobsterdata.com/}}. From this, minutely mid-market prices were constructed. More specifically, 30-minute log returns were then computed, leading to 13 observations per day (overnight returns were not considered). The raw data contained some erroneous values which were dealt with by setting any outliers, defined as a thirty minutely return of over 25 percent, to zero. 

% \MC{One sentence on why 1-2 bpts for costs - intraday, one can also execute passively and not necessarily cross the spread, but rather cash-in the spread if someone else crosses the spread and fills our order.}

\paragraph{Task}
The task is to predict the next 30-minutes market excess returns. Again, we adopt a rolling window backtesting framework with a fixed look-back of six months. Unlike the daily returns application, however, we do not recompute the model parameters at each backtesting step due to computational expense. Instead, we recompute the factors, loadings, and covariance matrix of the estimated idiosyncratic component every month. The universe of assets considered each month is defined as the set of assets for which there exists data at every time point of the look-back window for that month. If an asset enters the market mid-month, we include it in our universe at the beginning of the following month. We compare the same three models as in the application to daily returns prediction in Section \ref{subsec:daily}.

\paragraph{Performance metrics} 
As with the daily returns prediction, we consider the Sharpe ratio and mean daily PnL. Note, however, that in this case, $\text{PnL}_{t}$ is the sum of the intra-day returns for day $t$.
Figure~\ref{fig:intraday-CumPnL} shows the cumulative PnL in bpts for the three prediction models under an equal weighted portfolio allocation with no transaction costs. FNIRVAR yielded a Sharpe ratio of 2.20, the LASSO estimated factor plus sparse VAR model gave a Sharpe ratio of 1.14, whilst the static factor model alone gave a Sharpe ratio of 0.58. 
A passive strategy whereby one waits for an order to be filled without crossing the spread is realistic in an intra-day setting. As such, we do not consider fixed transaction costs here. In fact, for a passive strategy, we are likely to earn half-spread each time our order is filled, though one needs to take into account adverse selection as well.

\begin{figure}[t]
    \centering
    \includegraphics[width=0.975\linewidth]{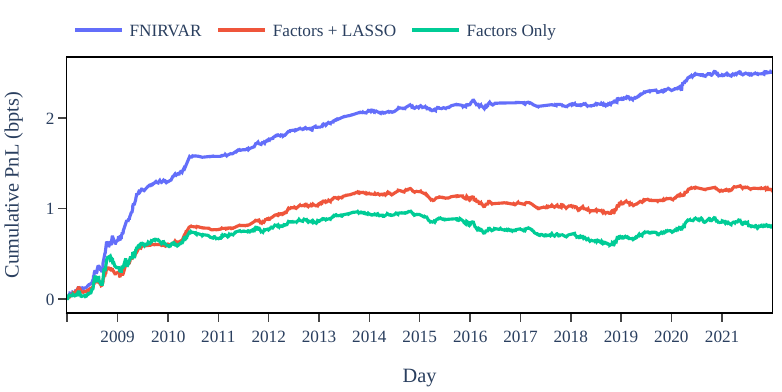}
    %%\vspace{-4mm}
    \caption{Cumulative PnL in bpts for three models backtest on 30 minutely intraday data.}
    \label{fig:intraday-CumPnL}
\end{figure}

%% file: Sections/FRED-MD.tex
\subsection{US Industrial Production}

\paragraph{Data description}
FRED-MD\footnote{Data available at: \url{https://research.stlouisfed.org/econ/mccracken/fred-databases}.} \citep{mccracken2016fred} is a publicly accessible database of monthly observations of macroeconomic variables, updated in real-time. We choose the August 2022 vintage of the FRED-MD database which extends from January 1960 until December 2019 ($T=719$), preprocessed as discussed in \cite{mccracken2016fred}. Only variables with all observations in the time period are used ($N=122$).

\paragraph{Task}
The task is one-step ahead prediction of the first-order difference of the logarithm of the monthly industrial
production (IP) index. We backtest each model from January 2000 until December 2019 using a rolling window
with a look-back window of 480 observations. We set $r = 8$ for the entire backtest as suggested by \citet{mccracken2016fred}. The order $l_{F}$ of the factor VAR model changed over the course of the backtesting period, having a mean value of 4.17 over the period. 

\paragraph{Peformance metrics}
Table \ref{tab:fred} shows that using FNIRVAR reduces the overall MSE compared with a factor model with no idiosyncratic model. In contrast, combining a factor model with a LASSO estimated sparse VAR model provides no forecasting improvement over the factor model alone. 

\begin{table}[t]
\scalebox{0.8}{
\begin{tabular}{cccc}
\toprule
 & FNIRVAR  & Factors Only & Factors + LASSO  \\
\midrule
Overall MSE   & \textbf{0.0074}  & 0.0077 & 0.0077 \\
\bottomrule
\end{tabular}
}
\caption{Overall MSE for the task of forecasting US IP.} 
\label{tab:fred}%
\end{table}

%% file: Sections/conclusion.tex
\section{Conclusion and discussion}\label{sec:coclusion}

In this work, we introduced a factor-driven network-based VAR model called FNIRVAR for modelling the common and idiosyncratic components of a high dimensional multivariate time series. In FNIRVAR, a static factor model is combined with a network VAR model whose coefficient matrix is the weighted adjacency matrix of a SBM. The regularisation due to the network makes FNIRVAR an appropriate model for high dimensional settings. Under the assumption of strong factors and a large eigengap, we estimate the factors and loadings via PCA. NIRVAR estimation is employed on the remaining idiosyncratic component leading to a block-sparse estimated VAR coefficient matrix. As the FNIRVAR estimator does not require an observed network, it is especially suited to applications in finance and economics where human-defined networks may not be available. 
To test the predictive performance of FNIRVAR, we constructed a backtesting framework for three datasets of daily returns, intraday returns, and macroeconomic variables. In all three cases FNIRVAR out-performed a static factor model as well as a static factor plus LASSO estimated sparse VAR model. FNIRVAR showed particularly compelling results on the financial applications, achieving a Sharpe ratio of $1.95$ for the task of daily returns prediction, and $2.20$ for the task of thirty minutely intraday returns prediction. These results highlight the benefits of including sparse estimation along with dense factor modelling. 

The FNIRVAR framework is especially suited to datasets containing groups of highly co-moving variables. From one perspective, the FNIRVAR estimator can be viewed as picking up block-weak factors \citep[see][for a discussion of block-weak factors]{barigozzi2024dynamic}. Theoretical analysis of the relation between FNIRVAR and weak factors is an interesting area for future research. Another future theoretical avenue is to investigate the asymptotic properties of the FNIRVAR estimator. \citet{fan2023bridging} develop theory on combining dense and sparse time series models; the challenge for FNIRVAR is to show that the idiosyncratic model is consistent. This would require extending existing SBM results \citep[for example,][]{lei2015consistency,chen2021spectral,cape2019signal} to the setting of dependent data. 
Also, the static factor model of FNIRVAR could be modified to allow for a dynamic factor model. In this case, dynamic PCA would be required to estimate the factor space \citep[see][for example]{barigozzi2024fnets}. Furthermore, the assumption of strong factors and a large eigengap could be replaced by a sparsity assumption on the idiosyncratic covariance matrix, with the optimisation-based approach of \citet{agarwal2012noisy} being used for estimation. 